\documentclass[showpacs,preprintnumbers, twocolumn]{revtex4-2}

\usepackage[colorlinks]{hyperref}

\usepackage[centertags]{amsmath}
\usepackage{amsfonts}
\usepackage{amssymb}
\usepackage{amsthm}
\usepackage{newlfont}
\usepackage{graphicx}
\usepackage{lipsum}
\usepackage{comment}

\newtheorem{thm}{Theorem}[section]

\newtheorem{defn}[thm]{Definition}
\newtheorem{cor}[thm]{Corollary}

\newtheorem{prop}[thm]{Proposition}
\newtheorem{example}[thm]{Example}

\DeclareMathOperator{\img}{img}

\begin{document}

\title{Towards Physics of Internal Observers: Exploring the Roles of External and Internal Observers.}
\author{Marcin Nowakowski}

 \email[Electronic address: ]{marcin.nowakowski@pg.edu.pl}
\affiliation{Faculty of Applied Physics and Mathematics, Gdańsk University of Technology, 80-952 Gdańsk, Poland}
\affiliation{National Quantum Information Center, 80‑309 Gdańsk, Poland}


\begin{abstract}
 In both quantum mechanics and relativity theory, the concept of the observer plays a critical role. However, there is no consensus on the definition of observer in these theories.
 Following Einstein's thought experiments, one could ask: \textit{What would it look like to sit inside a photon or to be a photon? And what type of observer could represent this more global perspective of the photon's interior?} 
 To address these questions, we introduce the concepts of \textbf{internal and external observers}  with a focus on their relationship in quantum theory and relativity theory.
 The internal observer, associated with the internal observables super-algebra, glues the  external interactions. 

 Drawing inspiration from the advancements in abstract algebraic topology, we propose mathematical representation of the internal observer. We also outline principles for ensuring the consistency of observers in terms of information theory. It becomes evident, through the analysis of the introduced \textbf{ hierarchy of observers}, that entanglement is a primitive of space-time causal relationships.

To further explore these ideas, we suggest studying the quantum cohomology groups of sheaves on simplicial complexes. The internal observer is linked with \textit{internal self-consistent loops}, as formalized through sheaf cohomology, and the stalks on the vertices of the superposed graphs are associated with the local quantum spaces of the external observer.

While external observers must abide by the relativistic causality linked with the no-signaling principle in quantum mechanics, the internal observer is inherently non-local and may be acausal. However, its consistency is maintained through the formulation of \textit{the self-consistency principle}.

 One of the goals of this paper is to construct the representation of the internal observer from the local external algebra of observables, which can be associated with external observers. Additionally, we demonstrate how the concepts of internal and external observers can be applied in the fields of quantum information theory, algebraic quantum field theory, and loop quantum gravity. The concept of internal observer seems to be fundamental for
further development of quantum gravity.

\end{abstract}

\maketitle

\section{Introduction}

The physics of observation refers to the study of how observations and measurements affect the physical systems being observed \cite{Measurement1, Measurement2, Measurement3, Measurement4}. In the context of quantum mechanics and relativity theory, this becomes a highly complex and nuanced field, as these theories offer different perspectives on the nature of reality and the relationship between the observer and the observed \cite{Obj1, Obj2, Obj3, Obj4, Measurement3}.

In quantum mechanics, the act of observation is seen as affecting the wave function of a quantum system, collapsing it from a superposition of states into a definite state. This is known as the observer effect and it raises fundamental questions about the nature of reality and the role of the observer in creating it. However, the observer's role is also controversial because there is no clear consensus on what constitutes an observer in quantum mechanics. Some argue that the observer must be a conscious entity, while others argue that any measurement apparatus can act as an observer \cite{Obs1, Obs2, Obs3, Obs4, Obs5}.

Relativity theory, on the other hand, highlights the observer-dependence of space and time. According to the theory of special relativity, the laws of physics are the same for all observers in uniform relative motion, and the speed of light is constant for all observers. This leads to a concept of relative observation, where the outcome of an experiment depends on the observer's state of motion.
An observer's measurements of space and time intervals are relative to their own frame of reference, which means that different observers can disagree on the sequence of events and the length of space and time intervals. This relative nature of observation is central to relativity theory and has profound implications for our understanding of the universe. This discussion is formulated in recent studies on the concept of quantum relativistic reference frames \cite{Zych2, Aharonov2, Aharonov3, Marletto4, Brukner} and on the quantum versions of the equivalence principle \cite{Zych1, Rosi, Marletto3, Brukner2, Marletto2, NowakowskiConsist}.

The intersection of quantum mechanics and relativity theory creates further complexities and challenges in the physics of observation. For example, the unification of quantum mechanics and general relativity is still an ongoing area of research and a major challenge in modern physics.

In conclusion, the physics of observation in the context of quantum mechanics and relativity theory highlights the observer-dependence of physical phenomena and the complex relationship between the observer and the observed.

These theories offer a radically different perspective on the nature of reality, in which the observer and the observed are inextricably linked, and where our understanding of the physical world is dependent on our methods of observation. As a result, the physics of observation remains an active area of research, with ongoing efforts to reconcile these theories and gain a deeper understanding of the underlying principles of quantum mechanics and relativity.

In the context of the discussion about observability and observers, the central question of this paper is how to represent an observer who is within an evolving physical system and how to represent observers from an external perspective. Using topological concepts, we aim to represent an observer associated with the interior of the system and those associated with the boundary of the system where measurements and interactions occur.

The meaning of realized and to be realized, connected with potentiality of physical phenomena, is inevitably connected with what was mentioned earlier.

We propose novel concepts of \textbf{internal and external observers} and aim to shed light on how to represent the internal perspective of a physical systems and the distinction between potentiality and realization in reality. This discussion is important for advancing our understanding of the fundamental principles of physics.

In this paper, we aim to establish a connection between the topology of the internal observer (similar to impossible figures) and the external observer. Our goal is to construct the representation of the internal observer from the local external algebras of observables, which can be associated with external observers and define principles ensuring consistency of these objects in purely information terms.

Inspired by the achievements of abstract algebraic topology, we present the candidates for mathematical representation of the internal and external observers. This paper is organized as follows:

Section \ref{secII} discusses the need for introduction of the internal observer concept in the context of quantum mechanics and relativity theory. 

Section \ref{secIII} introduces formal concepts of internal and external observers employing sheaf cohomology. In particular, it introduces a hierarchy of observers with which an algebra of observables can be associated at different levels of complexity. The principle of self-reconciliation of information for external observers, inspired by sheaf cohomology theory, is introduced to ensure consistency of information at the given level of the hierarchy of external observers. For the internal observer, we introduce the principle of self-consistency and show how the filtering mechanism leads to the lower level states from the state space of the internal observer. This discussion leads to the result that the higher level observers are more or equally non-local than the lower level observers. Finally, it is discussed that, in similarity to the hierarchy of observers, one can build a hierarchy of correlations up to super-entanglement of the internal observer.

In section \ref{secIV}, we discuss the concepts of internal and external observers in the context of loop quantum gravity, showing that external observers are directly linked with accessible sets of external observables on spin networks that bound the evolution of the system.

Section \ref{secV} considers the algebraic approach to quantum field theory (AQFT), which seems to be a natural choice in the context of the proposed hierarchy of observers, and shows how the internal observer can be represented in terms of the cohomology of infinite Lie algebras.
The discussion of the internal observer in a context of relativity theory, quantum mechanics, LQG and AQFT proves solid foundation behind this proposal.

In the appendix, we discuss an interesting proposal of the quantum principle of relativity \cite{Dragan} and how it can be translated into the language of external and internal observers. We also give a brief introduction to the consistency of CTCs (closed timelike curves), P-CTCs (postselected CTCs), and entangled histories, which motivates the introduction of the self-consistency principle for the internal observer. Finally, we present a technical overview of the theory of simplicial complexes and sheaf cohomology, and demonstrate that higher-order observers exhibit at least as much non-locality as their lower-level counterparts, if not more.

\section{The need for introduction of the internal Observer}\label{secII}

The subject of observation is central for the theory of general relativity and quantum mechanics. The revolution, brought to physics by Einstein, was built upon the special and general principle of relativity inspired by consideration of observation effects for different reference frames. On the other hand, quantum theory brought formalization of its central concepts in terms of quantum states, their evolution, and the measurement process where the observation act becomes a key aspect of reality.   

Einstein emphasizes \cite{AEinstein2} that the law of causality does not have the significance of a statement about the world of experience, except when \textit{observable facts} ultimately appear as causes and effects. These observable facts, as results of observation, account for the concept of events, which are the building blocks of space-time. The theory associates a set of variables $x_{\mu}$ with each point-event of space-time.

Einstein also states \cite{AEinstein2} that all verifications of space-time amount to the determination of space-time coincidences (understood as e.g. the interaction of scattering bodies). The introduction of a system of reference serves no other purpose than to facilitate the description of the totality of such coincidences. These considerations naturally lead to \textit{the general principle of covariance}, and space-time loses its physical objectivity as a structure \cite{AEinstein2}. Following Einstein's deep consideration of causality and the structure of space-time as a set of point-events with relations (formalized further by means of an introduced tensor metric $g_{\mu\nu}$), we arrive at the fundamental observation that space-time itself does not have any solid foundation without observers.

Quantum mechanics brings forth even more counter-intuitive phenomena related to the observation process, which is formalized in terms of algebras of observables $\mathcal{A}(O)$. A fundamental challenge in physics is understanding how to realize a physical process from a set of potential realities, which are formalized in quantum theory as quantum histories \cite{Griffiths, Cot1, MTSEH}. One can associate Feynman propagators with these quantum histories. 
They represent the probability amplitude for a particle to travel from one space-time point to another (Fig. \ref{EntPaths}). It is a mathematical expression that takes into account all the possible paths that the particle could take between the two points, in accordance with the principles of quantum mechanics.
However, the interpretation of these potential realities is still a big unknown in fundamental physics, and many secrets of this nature remain not fully understood. The controversies around an enigmatic phenomenon of the wave-function collapse are one manifestation of many not fully understood secrets of this nature. 

\begin{figure}[h]
\centerline{\includegraphics[width=10cm]{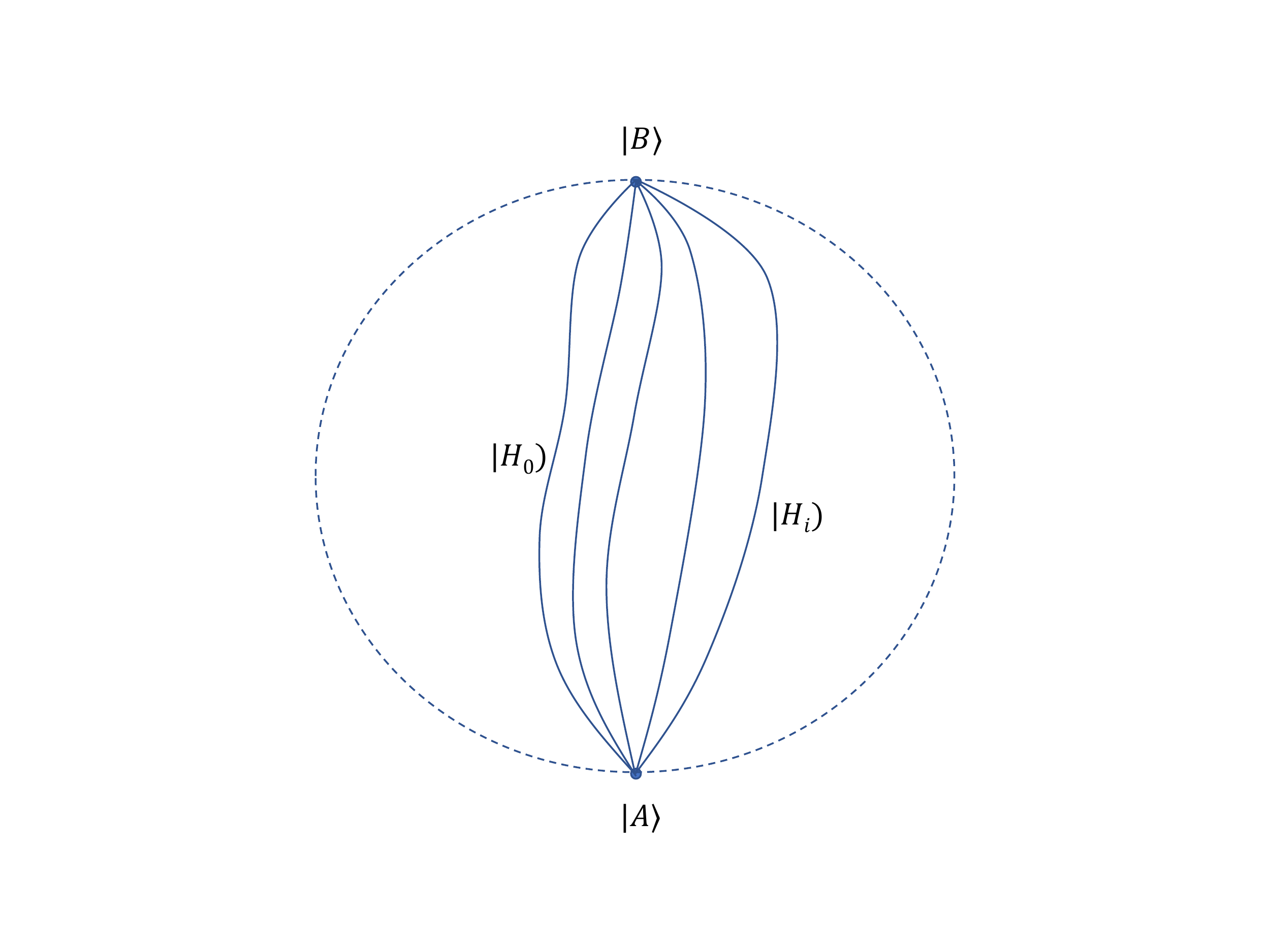}}
 \caption{
 Multiple paths connecting two space-time point-events $|A\rangle$ and $|B\rangle$ for a particle emitted at the point A in a state $|A\rangle$ and observed at B in a state $|B\rangle$. One can associate a quantum history $|H_i)$ and a probability amplitude $e^{i\phi_i}$ with each path, and the overall probability amplitude $\langle B | A\rangle \sim\sum_{i}e^{i\phi_i}$ with this quantum process.}
 \label{EntPaths}
\end{figure}

The problem of observing a physical process in quantum mechanics (QM) is related to its correlation with an external system, which is represented by a detector during measurement or another system during interaction. In some sense, the evolving physical object externalizes its state at the "moment" of observation. This external state is directly linked to the "realized" evolution path of a quantum system. After the measurement process, the state of a quantum system can be "collapsed," and the system can lose its so-called quantum coherence. What is not understood in the current state of physical theory, including beyond quantum theories, is the internal evolution of the system. It is as if the system could only exist through the measurement or interaction. There has been a long-standing debate, recalling the founders of quantum theory like Bohr \cite{Bohr}, about the physical existence of quantum objects between quantum jumps (e.g. electrons between discrete energy levels) or the existence \cite{Photon1, Photon2} of a photon in the arms of the Mach-Zehnder interferometer before it is measured or before it interferes with the surrounding field. 
Following Einstein's thought experiments, one could ask: \textit{What would it look like to sit inside a photon or to be a photon? And what type of observer could represent this more global perspective of the photon's interior?} This internal state cannot be observed directly by an external observer and its associated observables. Yet, it exists and is related to the inner reality of the photon.

In fact, the classical treatment of time is deeply intertwined with the Copenhagen interpretation of quantum mechanics, and thus with the conceptual foundations of quantum theory \cite{Obj1, Obj2, Obj3, Obj4}. All measurements of observables are made at certain instants of time, and probabilities are only assigned to such measurements \cite{Obj3, Obj4}.

In this context, we assume that space-time and other qualities of reality are constituted by the act of observation and observers. Yet, the concept of observers is still not fully understood.

To address these fundamental questions and consider the role of observers in the creation of reality, we propose to introduce two types of observers: \textbf{internal observers} and \textbf{external observers}, with further formalization in terms of algebraic structures associated with these objects. Intuitively, the internal observer should reflect the 'internal reality' or interior of the system $\mathcal{O}$ (i.e. $int\mathcal{O}$). The external observer should reflect what can be observed on the boundary of the system (i.e. $\partial\mathcal{O}$) by interaction with other systems. If a system does not interact, if it cannot be observed (measured), then it does not exist externally. Its external existence is defined by interaction and measurement which seems to be in agreement with how we understand quantum measurement process.

\begin{figure*}
\centering
\includegraphics[width=\textwidth]{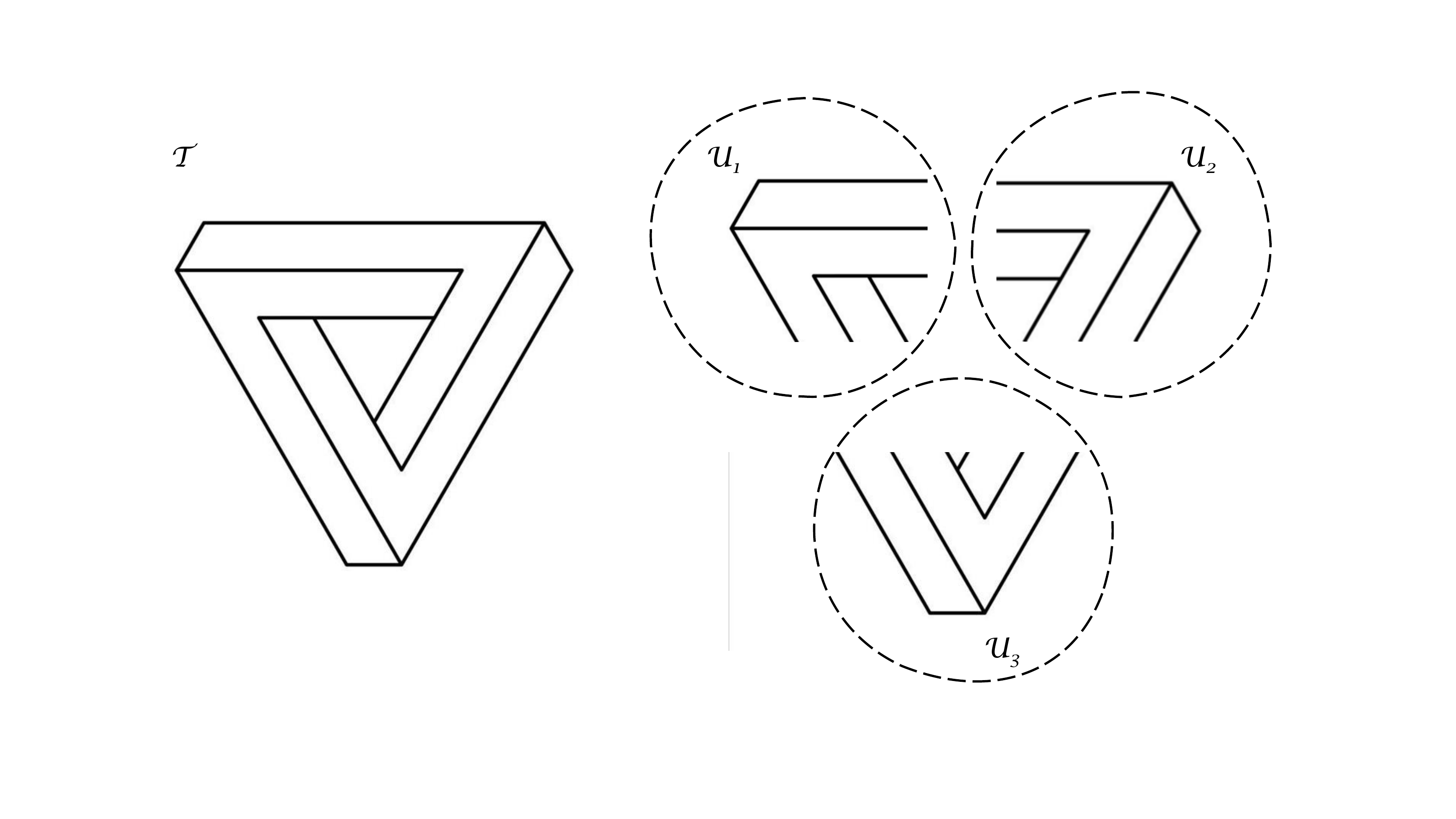}
 \caption{The Penrose tribar $\mathcal{T}$. The internal observer can be linked to a global impossible object that gives rise to local possible objects - external observers being still part of the overarching internal observer. $\mathcal{T}=\mathcal{U}_1\cup \mathcal{U}_2 \cup\mathcal{U}_3$ is the union of subsets $\mathcal{U}_1$, $\mathcal{U}_2$ and $\mathcal{U}_3$  which are topologically isomorphic to balls.}
 \label{Tribar}
\end{figure*}

Following Einstein's reasoning, which has been extended to local quantum observations \cite{Aharonov2, Zych1, Marletto2, Zych2, Rosi, Marletto3, Brukner}, a system of reference serves no other purpose than to facilitate the description of the totality of coincidences. The external observers operate in the causal world with consistent information about the observed facts. The assumption about the consistency of observed facts and causality forms the basis for further formulation of \textbf{the principle of self-reconciliation of information} for external observers. A variant of this principle is formulated in terms of the principle of general covariance, which should also hold for all quantum theories. In the model of reality presented in this paper, we claim that events, as building blocks of space-time, exist only in observers. Thus, the 'arena of space-time' external to observers is not well founded.

We assume that \textit{what exists is observable and what is observable can exist} or be realized. Thus, operationally we will consider further algebras of observables associated with 'local realities' of observers. The missing element is the internal 'global reality' of the evolving object that will be modelled by the concepts of an internal observer, its self-consistent internal space and the hierarchy of observers.

An internal observer is related to totality of possible observations and stores ultimately connections among the local 'external' realities.

We propose that \textit{time is not a fundamental concept for an internal observer}, since causality is directly related to point-events, which are elements of reality of the external observers. However, the internal observer ensures consistency of the overall evolution that will be formalized as \textbf{the principle of self-consistency of an internal observer}. The proposed principle of self-consistency is inspired by the consistency conditions for the CTCs (Closed Time-Like Curves) \cite{Deutsch}, P-CTCs \cite{Lloyd} and quantum entangled histories \cite{Cot1, MTSEH, NowakowskiAIP} (see \ref{Appendix2}). \textit{Those can be viewed as lower level variants in the hierarchy of observers of this principle for the internal observer}. These models provide a framework for understanding the behavior of quantum systems under different circumstances and ensuring that the models remain self-consistent and free of contradictions. 

We also argue that the concept of an internal observer precedes the concepts of time and space, which are generated by external observers for events and their causal relations. In this framework, causality holds for external observers but is not necessary for internal observers.

To support our statement about the lack of inherent causality for an internal observer, we should notice that an
internal observer without interaction with external observers does not generate any information relations among point-events, and therefore lacks any causal relations. This means that it is causally disconnected from external observations.

To provide a simple initial example that illustrates this discussion, let us again refer to Fig. \ref{EntPaths}.  An external observer is associated with the boundary of the evolving system (depicted as a dashed circle) where the preparation of the system in a state $|A\rangle$ and the observation of its state $|B\rangle$ happen. The internal observer is associated with the interior of the system (for simplicity) and this observer can be associated with the bundle of quantum histories $\{|H_i)\}$. Formal construction of quantum Hermitian observables is possible for the external observer on the boundary, as well as for an internal observable that is associated to the quantum history bundle.

\begin{figure*}
\centering
\includegraphics[width=\textwidth]{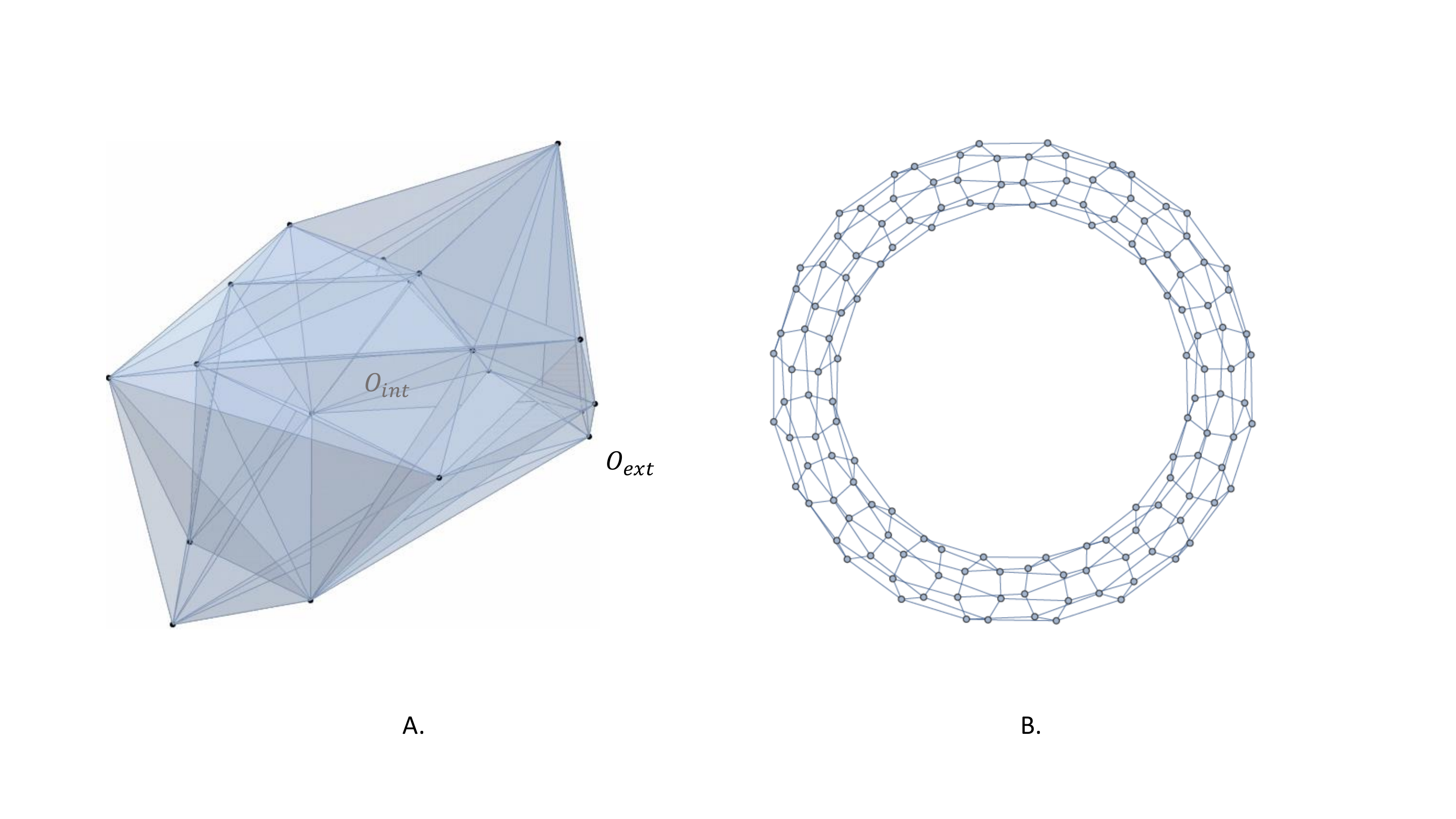}
 \caption{Examples of simplicial complexes: A) The internal observer $\mathcal{O}_{int}$ can be associated with the interior and the boundary of the left object. The external observer $\mathcal{O}_{ext}$ can be associated with the vertices in the simplified case, which are the realization point-events of the internal observer. B) The patch of torus with the simplicial complex structure.
 }
 \label{SimplicalComplexInt}
\end{figure*}

\section{Towards representation of the internal observer}\label{secIII}

In the search for formal representations of internal and external observers, we turn to inspirations from algebraic topology. One of the key inspirations is the concept of impossible topological objects, such as the famous Penrose tribar (Fig. \ref{Tribar}) or the Möbius strip \cite{Penrose}. The internal observer can be associated with a global impossible object that locally gives possible objects - external observers and the external observers remain part of the global internal observer. The key analogy relates to the difficulties observed in finding a consistent theory of quantum gravity, where local observations do not reflect the global aspects of reality \cite{GlobalReality1, GlobalReality2, GlobalReality3}. Another example of this is quantum non-locality, which can be observed in space-time for entangled multi-partite systems where local subsystems cannot detect that they are part of a global entangled state, and if measured, they would generate correct local distribution of measurement results. Recent research \cite{Cot1, MTSEH, NowakowskiAIP} suggests the existence of quantum entanglement in both space and time, which seems to be more fundamental than space-time itself. However, even though measurement statistics are available for measurement results, we cannot directly access the internal reality of the system that connects the measurement/interaction points. 

In this paper it becomes clear that \textit{classical, quantum and gravitational aspects of reality are the next levels of inherent hierarchy of observers and their realities, starting from local purely classical through local quantum to more global and non-local levels.}

To construct a formal representation of internal and external observers, we are covering the system with a simplicial complex (Fig. \ref{SimplicalComplexInt}). Thinking in terms of simplexes seems to be natural, especially for the boundary of the system, since any manifold (which is formally a set of points with connections) can be represented as a simplex. In a natural way, we would like to associate some spaces (in this case, algebras of observables) with these event points, which can be understood as points of interaction. The natural formal mathematical candidate for this is a consistent bundle of stalks which is a sheaf.
Thus, we propose to represent the concept of an internal observer $\mathcal{O}_{int}$ using quantum sheaves on simplicial complexes $\triangle$ (in particular graphs), where each stalk associated with a vertex represents the local reality of an external observer $\mathcal{O}_{ext}$ and can be formally represented by a local quantum manifold (e.g. a Hilbert-Schmidt space) (see \ref{Appendix3}). An important point is that each vertex can in general represent a nested graph (even generating a fractal graph) (see \ref{Appendix3}).

A vertex can be associated with a physical (realized) event, which is a primitive of the theory and serves as a realization point (an element of a space of point-events) of the internal observer. The node or event is a realization space that is equivalent, in this paper, to being "observed" or externalized due to interaction. Furthermore, nodes (events) are "places" of fibration of the internal observer, where reality buds. Thus, one could also visualize an internal observer object with a budding abstract tree.

To begin representing the internal and external observers, we need to review some concepts from sheaf cohomology theory \cite{Sheaf1, Sheaf2}, which are detailed in the appendix (see \ref{Appendix3}).

For a graph $\mathcal{G}=\{V,E\}$ and its vertices $v_i \in V$ and edges $e_i \in E$, it is worth mentioning that the vector spaces $\mathcal{F}(v)$ and $\mathcal{F}(e)$ are called \textit{stalks} of a sheaf $\mathcal{F}$. The sheaf can be interpreted as a bundle or a collection of stalks of data bound together by the underlying graph or the linear maps. Furthermore, the linear mappings $\mathcal{F}_{v\leftharpoonup e}$ linking objects from $\mathcal{F}(v)$ with the objects from $\mathcal{F}(v)$ encode the consistency condition that can be verified locally and which implies a sheaf structure of the spaces over the graph $\mathcal{G}$.   

If we associate with each vertex $v_i$ a local algebra of quantum observables $\mathcal{A}(O)$ then we still build a sheaf over a simplical complex. It is important to remember that formally the space of quantum hermitian operators (observables) is still a special vector space. The vertices of the simplical complex (of which a special case is a graph) are associated with the realization points by means of each interactions occur.  

\textit{As stated in the paper, we associate an algebra of observables with each level of complexity of the simplicial complex.} 

Assuming that local points of interaction (and the locally generated space-times) result from the more global structure of an internal observer, we can also look for algebraic features that persist at higher levels of complexity of an internal observer and its super-algebra of observables. Thus, naturally, we turn to (co)homology groups (see \ref{Appendix3}) and their persistence as a playground for our further considerations.

\subsection{Hierarchy of observers(ables)}

Our approach to representing both internal and external observers is based on information-theoretic concepts. These concepts are aligned with the fundamental principles of quantum mechanics, quantum field theory (QFT), and loop quantum gravity. However, they are viewed as even more fundamental, as they serve as the building blocks for the development of future theories of quantum gravity and quantum information.

By using information-theoretic concepts, the behavior of both internal and external observers can be described in a way that is consistent with the principles of quantum mechanics. This allows to better understand the role of information in the physical world and to develop new insights into the nature of reality.

The internal observer, denoted as $\mathcal{O}_{int}$ and associated with the super-algebra $\mathcal{M}_\infty$, serves as the object that glues observation events together due to its relation to infinite cohomology groups \cite{LieCoh1, LieCoh2, LieCoh3}.

We propose that the internal observer can be associated with a superspace $S_\infty$ consisting of connected topological spaces, including fractal spaces. This superspace can be represented as simplicial complexes, where each node represents a nested space. While these spaces can be of an abstract nature, or even regular space-time manifolds, we will focus on sheaves of local algebras of observables. The assumption of connectivity is due to the idea that \textit{the internal observer ultimately stores information connections within its internal space}. These intuitions shape reasoning about the algebraic structure of observables associated with this object.

To cover the internal observer topologically, multidimensional patches of local spaces are engaged. We construct its internal super-algebra through the sheaf ($\mathcal{F}$) cohomology group of observables on a simplicial complex $\triangle$ (see \ref{Appendix3}).

 \begin{defn}
The internal observer $\mathcal{O}_{int}$ is associated with an infinite sheaf cohomology super-algebra of observables $\mathcal{M}_\infty \cong H(\triangle,\mathcal{F})$.  
 \end{defn}

Using this terminology, one can associate a super-algebra $M_\infty$ and the filtration mechanism $\mathbf{F}$ (see \ref{Appendix3}) with the internal observer. After applying the $\mathbf{F}$ filtration to the k-level, one obtains a k-level simplicial complex (Fig. \ref{SimplicalComplexHierarchy}) with an associated algebra of observables that are accessible at that level of complexity:
\begin{equation}
   \mathcal{M}_\infty \rightarrow\cdots\rightarrow \mathbf{F}_k \mathcal{M}_\infty = \mathcal{M}_k  
\end{equation}
For the 1st-level of the hierarchy, one gets $\mathcal{M}_1(O)$ that can be associated with a quantum algebra of observables.

\begin{figure*}
\centering
\includegraphics[width=\textwidth]{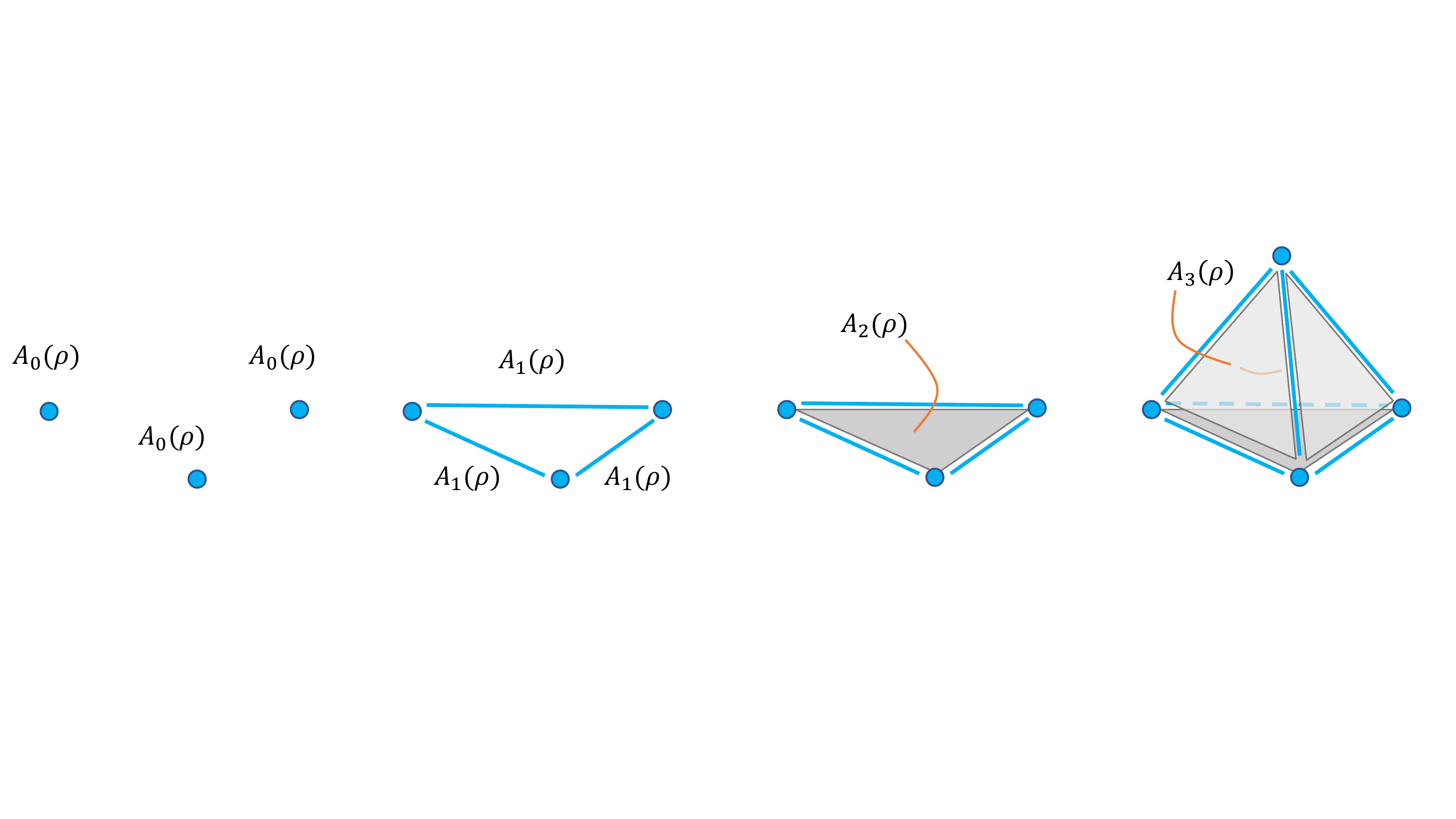}
\caption{The Hierarchy of Observers: Local algebras, denoted as $A_0$, are associated with each vertex of the 0-simplex, while non-local observables are associated with higher levels of the simplicial structure of observers, $A_1$ with edges, $A_2$ with faces and $A_3$ with the 'interior' of the 3-simplex. This hierarchy can be continued to higher-order simplexes. }
\label{SimplicalComplexHierarchy}
\end{figure*}

\begin{defn}
The external observer at the n-level of abstraction, denoted as $\mathcal{O}^{n}_{ext}$, is represented by a consistent collection of algebras $\mathcal{O}^{n}_{ext}: \mathcal{A} \rightarrow \mathfrak{M}(\mathcal{A})$ on a n-simplicial complex $\triangle_n$. Each vertex in the complex is associated with an algebra of accessible observables, denoted as $\mathcal{A}(O_{ext})$
\end{defn}
\textbf{At the 0th-level of the hierarchy}, every node can be considered as an individual external observer. The self-consistency of this representation is ensured by the principle of self-reconciliation of information among (external) observers.
As an example, one could consider a bipartite system AB in a maximally entangled quantum state $|\Psi\rangle=\frac{1}{\sqrt{2}}(|00\rangle+|11\rangle)$. At the 0th-level one can consider measuring one of its subsystems, e.g. A which is in a local mixed quantum state $\rho_{A}=\frac{1}{2}(|0\rangle\langle0|+|0\rangle\langle0|)$.  Further, we will discuss an example with such a system evolving in time which makes this structure more complex.

Analogous to sheaf cohomology, we assume that the information about the physical system at a given level of the hierarchy is consistent with the information at higher levels. Additionally, we assume that the information at a given level of hierarchy is consistent among connected external observers from this level. If we use the sheaves on simplicial complexes analogy, we require consistency in the sections over the observable sheaves.

These assumptions can be formally axiomatized as follows:
\begin{prop}
\textbf{The principle of self-reconciliation of information for external observers.} 
For any two observers $\mathcal{O}_1$ and $\mathcal{O}_2$ connected with an edge (an information connection) $E=\mathcal{O}_1:\mathcal{O}_2$:
\begin{equation}
\mathcal{F}_{O\leftharpoonup E}(\mathcal{O}_1)=\mathcal{F}_{O\leftharpoonup E}(\mathcal{O}_2)
\end{equation}
where $\mathcal{F}_{O\leftharpoonup E}:\mathcal{F}(E)\rightarrow \mathcal{F}(O)$ represents a map from the stalk associated with a node $O$ to the stalk associated with an edge $E$ for every connected pair $\{O,E\}$.
\end{prop}

The principle of self-reconciliation of information for external observers has its roots in the sheaf-theoretic approach to graphs. Based on the structure of the internal observer, one associates an algebra of observables with each node of the simplicial complex, and defines a stalk of a sheaf $\mathcal{F}$ for each node. By ensuring consistency via the consistency map, one guarantees that larger subsets of the object are consistent with finer subsets of the same subset \cite{Sheaf1, Sheaf2}.

In the case of the hierarchy of observers, it is required that the information collected at any level of the hierarchy is consistent with information at other levels. One can use the sheaf-theoretic approach to guarantee this consistency, just as it is done for the internal observer.

It is worth noticing that the assumption about consistency of information at a given level of observation is reminiscent of the general principle of covariance coined by A. Einstein for GR \cite{NowakowskiPrep}. 

It is also worth remembering, by applying the famous Gelfand-Neimark-Segal (GNS) representation theorem, that one can consider the algebra of observables abstractly without a particular operator realization. However, if one seeks an operator representation of the algebra, assuming that it is a *-algebra, one obtains a Banach space of bounded operators due to the isomorphism guaranteed by the GNS theorem. In particular, algebras of Hermitian observables can be associated with the nodes.

\textit{It is vital to observe that the local spaces associated with external observers can be space-time manifolds. However, the connections among them do not have to be of a spatio-temporal nature.} 

It is also observed that the higher the level of hierarchy for the observers, the more non-local observables can be associated with them. Thus, the internal observer can be represented as a globally non-local object. We can formally formulate this statement as follows:

\begin{cor}
The observers $\mathcal{O}_{ext}^{n+1}$ at the (n+1)-th level of the hierarchy ($\triangle_{n+1}$) are always more or equally non-local than the observers $\mathcal{O}_{ext}^{n}$ at the n-th level of the hierarchy ($\triangle_{n}$).
\end{cor}

This corollary can be formalized in the language of the associated algebras of observables on the (n+1)-simplical complexes and proved by induction (see \ref{Appendix4}).  The spatial and temporal quantum entanglement of multi-partite systems serves as a good example for considering the hierarchy of observers, from the classical level to local quantum subsystems and non-local multi-partite systems:

\begin{example}
Multipartite entanglement in space. 

As an example, consider a multipartite system $A_1A_2\cdots A_n$ in a quantum state $\rho_{A_1A_2\cdots A_n} \in \mathcal{B}(H_{A_1A_2\cdots A_n})$, for which the subsystem $A_k$ has a state represented by a reduced density matrix $\rho_k=Tr_{A_{i, i\neq k}}\rho_{A_1A_2\cdots A_n}$.

Naturally, if one assumes that $A_k$ is a qudit, then at the level of this single object, one cannot observe spatial quantum entanglement that may occur at the level of the whole multipartite system or its parts consisting of more than one qudit.
\end{example}

\begin{example} 
Quantum entanglement in time \& entangled histories \cite{MTSEH, NowakowskiAIP}.

Let us reconsider a spatial and temporal version of the $GHZ$-state.  A spatial GHZ-state on three qubits:
\begin{equation}
|\Psi_{ABC}\rangle=\frac{1}{\sqrt{2}}(|000\rangle + |111\rangle) \nonumber 
\end{equation}      
A temporal GHZ-state on one qubit:
 \begin{equation}
    |\Psi_{t_2 t_1 t_0})=\frac{1}{2}([0]\odot[0]\odot[0]+[1]\odot[1]\odot[1])\nonumber 
\end{equation}
Both states lead locally to mixed states e.g. $\rho_A=\frac{1}{2}(|0\rangle\langle 0| + |1\rangle\langle 1|)$ (which represents noise) and $\rho_{t_0}=\frac{1}{2}(|0)( 0| + |1)(1|)$, but also any chosen pair is in a separable state, i.e. $\rho_{AB}=\frac{1}{2}(|00\rangle\langle 00| + |11\rangle\langle 11|)$. However, both systems keep non-local nature globally. 

Furthermore, while a three-qubit spatial state $|\Psi_{ABC}\rangle$ cannot be simply extended to an $n$-qubit state due to a lack of additional resources, a temporal state can always be extended to $n$ times (reaching infinity as a limit if not constrained by any discrete limit like, for example, a Planck time) between the constrained past $t_{0}$ and future $t_{2}$, as far as the time steps are defined for the external observer.

\begin{eqnarray}
    |\Psi_{ABC}\rangle & \nrightarrow & |\Psi_{ABCC_1\ldots}\rangle\\
    |\Psi_{t_2 t_1 t_0})&\rightarrow &|\Psi_{t_2\ldots t_1 t_{0n}\ldots t_{01} t_0})
\end{eqnarray}
This tricky feature of temporal correlations is one of the reasons of polyamoric nature of time \cite{MTSEH, Grudka} as presented further for an evolution of AB system with 'injected' Pauli measurements in a pre-selected state $|\Phi_{+}\rangle$. The action of the unitary Pauli operations (Fig. \ref{GlobalHistory}) on the subsystem A in between the pre-selected $|\Phi_{+}\rangle$ and the post-selected $|\Phi_{+}\rangle$ can be represented as a history $|H_{\Gamma})$:
\begin{equation}\label{Pauli}
|H_{\Gamma})=[\Phi_{+}]\odot[\Phi_{+}]\odot[\Phi_{+}]\odot[\Phi_{+}]\odot[\Phi_{+}]
\end{equation}
with bridging operators $U(t_{1},t_{2})=\sigma_{x}\otimes I_{B}$, $U(t_{2},t_{3})=\sigma_{y}\otimes I_{B}$, $U(t_{3},t_{4})=\sigma_{z}\otimes I_{B}$ 
and $U(t_{4},t_{5})=I$ with post-selected state $|\Phi_{+}\rangle$. 
Interestingly, if we ask now for a history of the subsystem $A$ in this evolution, we get a temporal version of an entangled GHZ state:
\begin{equation}
|H_{A})=\frac{1}{2}([0]\odot[0]\odot[0]\odot[0]\odot[0]+[1]\odot[1]\odot[1]\odot[1]\odot[1])
\end{equation}
with corresponding evolution $U_{A}(t_{1},t_{2})=\sigma_{x}$, $U_{A}(t_{2},t_{3})=\sigma_{y}$, $U_{A}(t_{3},t_{4})=\sigma_{z}$ 
and $U_{A}(t_{4},t_{5})=I$
where $|H_{A})$ is derived from the global $|H_{\Gamma})$ tracing out $B$ party over all times (\cite{MTSEH}) keeping consistency of the derived reduced evolution of the subsystem (note that it is not a mere analogy of spatial trace out operation over all times, the binding evolution between the time steps has to be kept). 

\begin{figure}[h]
\centerline{\includegraphics[width=9cm]{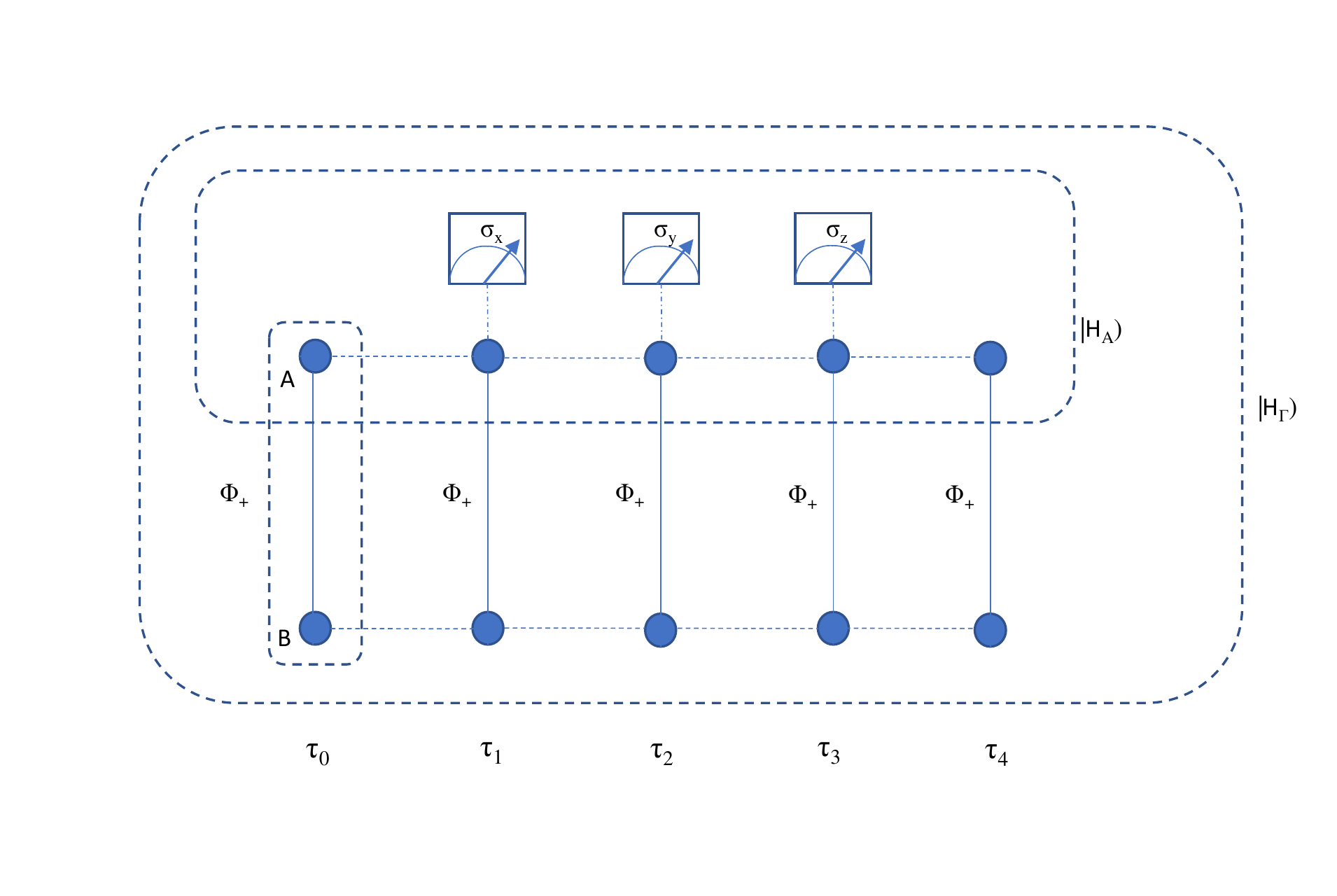}}
 \caption[Evolution]{$|H_{\Gamma})$ represents the multi-time global evolution of the bipartite system $AB$, pre-selected and post-selected in the state $\Phi_+$ with consecutive measurements on subsystem $A$. One can obtain an entangled temporal version of the GHZ state for a history $|H_A)$ from a global history $|H_{\Gamma})$. }
 \label{GlobalHistory}
\end{figure}

This local evolution can be equivalently represented by the underlying quantum structure with a trivial evolution:
\begin{eqnarray}\label{localev}
|H_{A})&=&\frac{1}{2}([0]\odot\sigma_{z}[0]\sigma_{z}\odot\sigma_{y}[0]\sigma_{y}\odot\sigma_{x}[0]\sigma_{x}\odot[0]\nonumber\\
&+&[1]\odot\sigma_{z}[1]\sigma_{z}\odot\sigma_{y}[1]\sigma_{y}\odot\sigma_{x}[1]\sigma_{x}\odot[1])
\end{eqnarray}

\end{example}

The above global and local history is constructed without a definite state of the local system A at intermediate times. However, if the local state of A is read and a definite state $|+\rangle$ or $|-\rangle$ is obtained for a specific measurement setting, then the history becomes disentangled and effectively, the considered history of the process is separable. It is important to note that measurement itself is an intrinsic part of the process and contributes to the creation of the particular history.

What appears to be the most complex and not fully understood aspect of the measurement process in modern quantum physics is the so-called collapse of the wave function. Based on the reasoning in this paper and the relativity of observation for external observers, we conclude that the collapse is a local process that is accessible to external observers at the lower levels of hierarchy (for local quantum and local classical systems when the measurement results are recorded classically). However, the entire information is still stored in the internal observer at the higher, non-local levels of hierarchy. It can be associated formally with the entangled evolution paths for spatial and temporal correlations.    

\subsection{Self-consistency of an internal observer}

While the principle for external observers is applicable to all levels of the hierarchy of observers, the principle of self-consistency of an internal observer is of a global nature. There is some analogy to the concept of quantum non-locality in space-time. If one considers a multipartite quantum states entangled in space and time, it leads to consistent sets of local classical probability distributions, however, locally at the classical level the phenomenon of quantum entanglement does not exist although the probability distributions can violate Bell-inequalities. If one gets at the higher level of the hierarchy, then can get to spatial non-local effects still not finding any signs of non-locality in time. Thus, in general the internal observer has to keep self-consistency, independently from interaction effects, at the global fundamental level which, as proposed in this paper, is of information nature. The proposed principle of self-consistency is inspired by the consistency conditions for CTCs (Closed Timelike Curves) by Deutsch \cite{Deutsch}, P-CTC (Post-selected CTCs) by Lloyd et al. \cite{Lloyd} and quantum entangled histories \cite{Griffiths, Cot1, MTSEH, NowakowskiAIP} which ensure causal consistency (\ref{Appendix2}). 

If we associate multi-dimensional sheaves on simplicial complexes with an internal observer $\mathcal{O}_{int}$, we can ask questions about self-consistency of the information loops (that can be immediately translated into the language of co-cycles in case of simplicial complexes). Thus, we consider the following representation chain:
\begin{equation}
\mathcal{O}_{int}\longrightarrow\triangle\longrightarrow\mathcal{M}_{\infty} \longrightarrow \mathbf{H}(\triangle,\mathcal{F})
\end{equation}
where $\triangle$ stands for a simplical complex, $\mathcal{M}_{\infty}$ for the infinite super-algebra of observables associated with the internal observer and $\mathbf{H}(\triangle,\mathcal{F})$ stands for sheaves cohomology over $\triangle$. 
 
\begin{prop}
\textbf{The principle of self-consistency of an internal observer.} The state of an internal observer $\mathcal{O}_{int}$
can be represented as a 'superposition' of all possible consistent loops (co-cycles) over its internal space $\mathcal{S}_\infty$:
\begin{equation}
|\Lambda\rangle\rangle=\sum_{\circlearrowleft\in\triangle} e^{i\phi(\circlearrowleft)}|\circlearrowleft\rangle\rangle
\end{equation}
where each loop is associated with an internal state $|\circlearrowleft\rangle\rangle  \in \mathcal{S}_\infty$, meeting the self-consistency conditions:\\
(*) $\langle\langle\Lambda|\Lambda\rangle\rangle=\mathbb{I}$,\\
(**) $\exists_{\alpha_i\in \mathbb{C}}\sum_i\alpha_i|\circlearrowleft_i\rangle\rangle=\mathbb{I}$, (decomposition of identity into a sum of self-consistent loops with the assumption of existence of a countable set of internal loops $\{|\circlearrowleft_i\rangle\rangle\}$).
\end{prop}

We shall introduce also a mapping from the internal state to the external observers representation, which can be understood as mapping from the interior to the boundary of the internal observer:
\begin{equation}
    \partial_k: |\circlearrowleft\rangle\rangle\rightarrow |\circlearrowleft\rangle_k
\end{equation}
where $|\circlearrowleft\rangle_k$ denotes representation of the state for the external observer at $k$th-level of the hierarchy (in relation to the k-simplex $\triangle_k$).
Following the QFT connotations \cite{}, we can call $|\Lambda\rangle\rangle$ the ground state of the internal observer and $|\circlearrowleft\rangle\rangle$ states associated with internal information loops (knots).  

In this picture of the internal observer, it is a static object consisting of all possible information loops over its structure which are self-consistent. In natural way, one can generate a dynamic representation traversing this object at different levels of the  observation hierarchy.  The loops can be associated to consistent co-cycles belonging to sections over the sheaves on the simplicial complex. Apparently, the consistency is ensured by higher-level entanglement of $|\Lambda\rangle\rangle$. We are calling that \textit{super-entanglement} since it can lead to a lower level entanglement when the boundary (or projection) of the state is performed to $\triangle_2$ or $\triangle_1$ (see \ref{Appendix4}).

Following the concept of hierarchy of observers and the concept of the internal super-observer we can also build a similar hierarchy for entanglement ensuring consistency of information distributed in this structure:
\begin{equation}
    E_{int}\rightarrow E_k \cdots\rightarrow E_{q (quantum)}\rightarrow E_{0(classical)}
\end{equation}

Typical spatial or temporal quantum entanglement is lower-level projection from the higher dimensional object.
\textit{Causality is specific for external observers since it is a feature associated with events on the boundaries of the internal observer.}

\textit{The higher dimensional super-connections in the internal observer are a predecessor of entanglement at lower levels of the observers' hierarchy, which is the predecessor of classical causality.} All this phenomena exist at once within the internal observer object but are characteristic for its different sections (or filtration levels).
Again, we can conclude that the so-called collapse of wave function is the phenomenon at the lowest levels of the hierarchy but it does not destroy consistency of the internal observer and does not collapse its existence, which brings also a new perspective on the measurement process itself \cite{NowakowskiPrep}.

\section{Loop quantum gravity perspective}\label{secIV}

Loop Quantum Gravity (LQG) \cite{C.Rovelli, LQG1, LQG2, LQG3} is a proposed theory of quantum gravity that attempts to reconcile Einstein's theory of general relativity with quantum mechanics. In LQG, the geometry of spacetime is quantized, and the fundamental objects are loops and networks of loops called spin networks.

The concept of an internal observer can be incorporated into LQG by associating a spin network with each vertex of the simplicial complex representing the internal observer. Each spin network represents the quantum geometry of the internal observer at that vertex. The edges of the spin network correspond to connections between the vertices of the internal observer, which in turn correspond to possible interactions between the internal observer and its environment.

Furthermore, the algebra of observables associated with each level of complexity of the simplicial complex can be represented as a set of operators acting on the spin networks. These operators correspond to physically observable quantities in the internal observer's reality.

Overall, the incorporation of the concept of an internal observer into LQG provides a new perspective on the theory, allowing for the study of the interplay between quantum gravity and observer-dependent phenomena.

In this section, I delve into the concept of internal and external observers in the context of transition amplitudes in quantum gravity. It's worth noting that any measurement involving the gravitational field is inherently a quantum gravitational measurement, as highlighted by C. Rovelli \cite{C.Rovelli}. To calculate transition amplitudes, we use 'spin networks' or quantum spin networks, which are background-independent and represent the structure of space-time itself. The language of spin networks is also naturally suitable for expressing the path-formulation of relativistic quantum field theories, making it a valuable tool for studying the behavior of internal and external observers in the context of quantum gravity.

C. Rovelli has stated \cite{C.Rovelli} that the main accomplishments of covariant LQG to date are as follows: (i) The amplitudes are finite at every order; (ii) At each order, the amplitudes demonstrate a well-defined classical limit that corresponds to a truncation of classical general relativity; and (iii) The theory has been expanded to include fermions and Yang-Mills fields.

The general form of a propagator depends on the fluctuating metrics of gravitational fields. However, the fundamental idea of constructing paths for the evolution of a system is essentially the same as the path integral approach used by R. Feynman for quantum systems. Therefore, we can still enhance our understanding by employing the formalism of internal and external observers.

One can consider the transition amplitude for the process starting with metric $g_{i}$ at time $t_{i}$ and finalizing with a metric $g_{f}$ at time $t_{f}$. The choice of time-indexes is not crucial due to the assumption of diffeomorphism-invariance of the fields. One can immediately observe that what matters is the relation between the indices for the external observers. The propagator gets the following form:

\begin{equation}
    \mathcal{A}(g_f, g_i)=\int_{g_{i},t=0}^{g_{f},t=0} D[g_{\mu\nu}(x)]e^{iS_{GR}(g)}
\end{equation}
where $S_{GR}(g)$ stands for the general relativistic action of the path between the initial and final state. 
One of the unsolved problems of quantum gravity is a comprehensive understanding of the techniques of construction of the measure $D[]$. 
Assuming that physical space-time is discrete, loop quantum gravity sheds some light on this problem by simplifying the calculation of propagators in many cases. It is important to note, as pointed out by C. Rovelli \cite{C.Rovelli}, that $x$ does not represent the classical variable, but instead labels eigenstates of the variable, specifically, eigenstates of the 3-geometry. Because of the discreteness of space-time, the propagator in quantum gravity must be a function of spin networks \cite{C.Rovelli}.  

As a consequence, the propagator mentioned above can be expressed as a function of the initial spin network state $|s_{i}\rangle$ and the final spin network state $|s_{f}\rangle$, which are defined by the boundary conditions set by the external observers for the history of the system. There are also intermediary states of the spin network that naturally arise. The probability amplitude is determined by summing over the histories of the spin networks:
\begin{equation}
    \mathcal{A}(\sigma_f, \sigma_i)= \sum_{\sigma}\mathcal{A}(\sigma) 
\end{equation}
where $\sigma$ is again a discrete history of the spin networks: $\sigma=(s_{f}, s_{n}, \ldots, s_{1}, s_{i})$ and the probability amplitude associated with a single history is: $\mathcal{A}(\sigma)=\Pi_{l} \mathcal{A}_l({\sigma})$ (with $l$ indexing the history steps for the external observer). The history of the spin-network is a so-called  spinfoam and different models for its representation have been proposed \cite{}. In analogy to the quantum path integrals, one can state $\mathcal{A}(\sigma)=\langle s_{i+1}|e^{-\int d^3 x H(x) dt} |s_{i}\rangle_{\mathcal{K}}$ for the evolution between steps $|s_{i}\rangle$ and $|s_{i+1}\rangle$. We can now state that state of the individual history of a spin-network can be represented, in analogy to entangled histories formalism, as:
\begin{equation}
    |s)=[s_{f}] \odot [s_{n}]\odot \ldots \odot [s_{1}]\odot [s_{i}]
\end{equation}
and furthermore, we can construct superpositions of the spin network histories (spinfoams):
\begin{equation}
    |\mathcal{H}_s)=\sum_j \alpha_j [s^{j}_{f}] \odot [s^{j}_{n}]\odot \ldots \odot [s^{j}_{1}]\odot [s^{j}_{i}]
\end{equation}
The evolving spin-network can be represented as a specific graph $\Gamma$ evolving with each step (that can be interpreted as a time flow for the external observer or an external step-index) where changes are governed by the Hamiltonian $H$ acting on nodes of the graph. 
\begin{widetext}
\begin{equation}
\mathcal{A}(\sigma, \sigma')= \sum_{\partial\sigma=\sigma \cup \sigma'}w(\Gamma(\sigma))\prod_f\dim (j_f)\prod_e \mathcal{A}_e(j_f, i_e)\prod_v \mathcal{A}_v(j_f, i_e)
\end{equation}
\end{widetext}
where $\mathcal{A}_v(j_f, i_e)$ represents the vertex amplitude, $\mathcal{A}_e(j_f, i_e)$ represents the edge amplitude and $w(\Gamma(\sigma))$ is a weight factor. As a matter of convention, $v$ stands for the vertices of the graph where branches of the evolution meet, $e$ stands for the edges of the graph being the wordlines of the vertices and $f$ for faces of the multidimensional graph.  
This sum is over a spinfoam $\sigma$ bounded by spin networks $\sigma$ and $\sigma'$. In this case, we can associate the boundary with the external observer imposing initial and final conditions for the evolution or alternatively, interacting with the internal observer of the spinfoam at the boundary. 

Again the transition amplitude for the spinfoam is defined by means of the sum over histories between the bounded two states of the gravitational field represented by $s$ and $s'$ which can be associated with states of the external observer. 

It is worth mentioning that with information about faces f, edges e and vertices v, one can reconstruct the graph $\Gamma$ representing information about the evolving physical object (in this specific case being a region of space-time). From this perspective graphs are mathematical structures representing data sets accessible to observers. 

In this case \textit{external observers} are directly linked with accessible sets of external observables on spin networks bounding evolution of the system: 
\begin{equation}
    \mathcal{O}_{ext}=\sum_{j}|\sigma^{j}\rangle\langle \sigma^{j}|
\end{equation} 
One can regard the metric tensor $g_{\mu\nu}$ as an object that encodes spatio-temporal information from the perspective of external observers. However, for internal observers, one can associate an internal tensor metric $\eta_{\mu\nu}$ (for the internal space $\mathcal{S}_\infty$) that encodes the observer's own spatio-temporal information. At branching points, where events occur from the perspective of external observers, the internal metric is equal to the external metric.

One considers histories of spin networks (in QFT one considers field histories) but one should follow the question: where is the information stored for the histories and by means of which physical reality is constituted? This questions is naturally related to information and energy of the physical system which if bounded, leads to bounded bulk of potential histories. To simplify this considerations, we employ \textit{the internal observer which is linked with totality of information stored within}. The external observer can be represented as a boundary of the internal observer.

We would like to mention that the concepts discussed in this paper do not exhaust the subject of observation in quantum gravity but highlight the problem of realization which seems to be central for definition of evolution and energy functionals in the future theory. By using this approach, we can gain a deeper understanding of the relationships between the physical system, the information it contains, and the role of observers in shaping our understanding of these phenomena.

\section{The algebra of QFT observables}\label{secV}

In this section, we will review the concepts of internal and external observers, including their hierarchy, within the context of algebraic quantum field theory (AQFT) and its core postulates. Readers who are interested in a comprehensive elaboration of this approach to QFT are invited to consult the AQFT literature \cite{DHR1,DHR2,DHR3,DHR4,DHR5}.

This approach to QFT is also known as the Haag-Kastler axiomatization of QFT \cite{HKastler} and involves associating unital $*$-algebras of observables to regions $\mathcal{R}$ of spacetime, i.e. $\mathcal{R}\rightarrow \mathcal{A}(\mathcal{R})$, with the assumption, inherited from the theory of relativity, that space-like separated regions are associated with commuting algebras of observables. Specifically, for any two causally disjoint regions $\mathcal{R}_1$ and $\mathcal{R}_2$, all observables $O_1 \in \mathcal{A}(\mathcal{R}_1)$ and $O_2\in \mathcal{A}(\mathcal{R}_2)$ have a vanishing commutator, i.e. $[O_1, O_2]=0$ (the Einstein causality). It is also required that whenever $\mathcal{R} _1 \subset \mathcal{R}_2$, then $\mathcal{A}(\mathcal{R}_1) \subset \mathcal{A}(\mathcal{R}_2)$ (the isotony condition), meaning that the local algebras create a nested algebraic structure. These algebras also satisfy the Poincaré covariance and the existence of dynamics axioms. Specifically, if $\mathcal{R}_1 \subset \mathcal{R}_2$ and $\mathcal{R}_1$ contains a Cauchy surface $\Sigma$ of $\mathcal{R}_1$, then $\mathcal{A}(\mathcal{R}_1)=\mathcal{A}(\mathcal{R}_2)$. The Poincaré covariance means that symmetric transformations $\rho$ of the local regions of spacetime can be associated with the automorphisms $\alpha(\rho):\mathcal{A}(\mathcal{R}_1)\rightarrow \mathcal{A}(\rho\mathcal{R}_1)$. This structure of algebras for a nested family of regions $\mathcal{M}=\bigcup_i \mathcal{R}_i$ can be used to reconstruct the whole spacetime algebra $\mathcal{A}(\mathcal{M})=\bigcup_i \mathcal{A}(\mathcal{R}_i)$. To construct correlation functions, it is also necessary to define the functional $\langle\cdot\rangle :\mathcal{A} \rightarrow \mathbb{C}$ (where $\langle\mathbb{I}\rangle=1$), which is used to calculate the expectation values for the algebra observables. With the above defined, one can construct, for example, a Hilbert space spanned by the states $|\psi(\hat{x})\rangle=\phi(x_1)\cdots \phi(x_n)|0\rangle$ for the scalar field $\phi(x)$.

The algebraic axiomatization of QFT is a framework that has been extended and modified in various ways to address specific physical phenomena, including Conformal Field Theory and other variations. Despite these extensions, the core concepts of the framework remain unchanged, and can be seen as an attempt to reconcile quantum theory with the relativistic understanding of space-time.

However, the challenge of reconciling classical, well-defined and stable notions of space-time with quantum concepts is a major hurdle in formalizing QFT within this framework. This is because in the algebraic approach, space-time is seen as a collection of quantum systems, and the notion of a classical, well-defined space-time emerges only in the classical limit. As a result, the formalization of QFT within the algebraic approach requires a careful treatment of the interplay between quantum and classical concepts, and this is an active area of research in the field.

In the context of this paper, it is apparent that \textit{ it is an attempt to represent local (external) observers within a well-defined space-time} without exploring an internal observer representation, which is inherently non-local and not defined in terms of space-time. It should be noted that the nested algebra of observables attempts to reconstruct the internal observer (in analogy to quantum tomography \cite{Tom1, Tom2}) and is therefore necessarily incomplete.   
 
Actually, Poincare covariance justifies thinking about space-time in terms of information by associating abstract observable group elements with transformations of space-time as the arena of physical phenomena. This is based on the principle of self-consistency of information for external observers.

Building on the 'Haag-Araki' representation of QFT algebras, it can be assumed that all local algebras ${\mathcal{A}(\mathcal{R}): \mathcal{R} \subseteq
\mathcal{M}}$ are von Neumann algebras that act on some Hilbert space $\mathcal{H}$, and that there exists a translationally invariant vacuum state $|\Omega\rangle \in \mathcal{H}$. Consequently, measures of quantum entanglement, which are commonly used to study entanglement in discrete quantum systems \cite{EM1}, can be applied to measure quantum entanglement among sub-regions of $\mathcal{M}$. This research area is rapidly growing and gaining more attention recently \cite{EntAQFT1, EntAQFT2, EntAQFT3}.

The network and hierarchical structure of observers defined in this paper also applies to AQFT. AQFT assumes that all physical information can be reconstructed from the observables network, which is its fundamental object, and the connection with the Hilbert representations can be derived through the choice of the algebraic state. For instance, in the aforementioned QFT on Minkowski space-time, the vacuum state serves as the algebraic state.

The algebraic formulation of QFT also incorporates the concept of superselection sectors. Physical observables can be represented as block diagonal operators on a Hilbert space $\mathcal{H}$ with a decomposition $\mathcal{H}=\bigoplus_{i\in I}\mathcal{H}_i$, where $i$ corresponds to the superselecting index and $\mathcal{H}_i$ represents the superselection sector. For example, in a theory with charge, the observable $\mathcal{Q}$ is conserved during any interaction in the measurement process, and has a discrete spectrum with eigenvalue equation $\mathcal{Q}|\Psi\rangle=q |\Psi\rangle$. By introducing a zero-one projective measurement $P$, we obtain $\mathcal{Q}P|\Psi\rangle=qP |\Psi\rangle=P\mathcal{Q}|\Psi\rangle$, which implies $[\mathcal{Q},P]|\Psi\rangle=0$ due to charge conservation. The charge eigenspaces can then be regarded as superselection sectors of the global Hilbert space.

Superselection sectors correspond to unitarily inequivalent representations of the algebra of observables. A notable contribution of AQFT \cite{DHR1, DHR2, DHR3, DHR4, DHR5} is that superselection sectors correspond to irreducible representations of a compact Lie group $\mathcal{G}$ \cite{Hilgert, Stasheff}, which is the global gauge group. The local algebra of observables $\mathcal{A}(O)$ is uniquely derived from the net, and consists of gauge-invariant fields $\mathcal{F}(O)$ (which can include smeared observables such as Dirac fields). 
 
One promising approach for constructing a simplified representation of an internal observer superalgebra involves building on the cohomology of infinite-dimensional Lie algebras $L_\infty$, along with additional principles that have been developed in previous sections of this research.

In particular, the study of infinite-dimensional Lie algebras in quantum field theory \cite{LieCoh3} has yielded important insights. For example, the affine Kac-Moody algebras and the Virasoro algebra have been extensively studied and have proven to be powerful tools in the analysis of conformal field theories and other physical phenomena.

The cohomology of $L_\infty$ \cite{LieCoh1, LieCoh2, LieCoh3, LieCoh4} is particularly interesting in this context, as it provides a framework for understanding the algebraic structure of observables in QFT. By analyzing the cohomology classes associated with $L_\infty$, it is possible to gain a deeper understanding of the symmetry groups that underlie physical phenomena, and to explore the mathematical connections between seemingly disparate physical systems.

A Lie algebra $(\mathfrak{g},[\cdot])$ consists of a vector space $\mathfrak{g}$ and a bilinear map $[\cdot]:\mathfrak{g} \wedge\mathfrak{g}\rightarrow \mathfrak{g}$, satisfying the Jacobi identity. As an example, one can consider any associative algebra $\mathcal{A}$ with the commutator $[\mathbb{A},\mathbb{B}]=\mathbb{A}\mathbb{B}-\mathbb{B}\mathbb{A}$ giving it the structure of a Lie algebra. In general, associative quantum observable algebras have the structure of Lie algebras.

As an extension to graded tensor vector spaces, we can introduce an infinite Lie algebra $L_\infty$ as an algebra consisting of a $\mathbb{Z}$-graded space and a collection of linear maps $[\cdot]:\mathfrak{g} \wedge^k\mathfrak{g}\rightarrow \mathfrak{g}$ of degree $2-k$ called $k$-brackets, satisfying the generalized Jacobi identity. A degree of zero means that the space is ungraded. There exist various equivalent definitions of infinite Lie algebras \cite{LieCoh3}, with extensive studies of their applications in QFT and the standard model.

In the context of observers and their representations, it is important to also consider the concept of cohomology of Lie algebras. Similar to sheaf cohomology on simplicial complexes, we define Lie algebra cohomology as $\mathbf{H}^k(\mathfrak{g},F)=\frac{\ker\delta^k}{\operatorname{im}\delta^{k-1}}$ (see \ref{Appendix}), where $F$ stands for the representation of $\mathfrak{g}$. This cohomology group describes the obstruction to the existence of an $F$-valued k-cocycle that satisfies a certain condition and are discussed in the section about internal loops of the internal observer. This cohomology group provides information about the algebra of observables and their representations, and can be used to study the properties of an internal observer superalgebra.

The cohomology of Lie algebras is a tool used to study the algebraic structure of Lie algebras and their representations. In the context of quantum field theory, the cohomology of $L_{\infty}$ can be used to study the internal observer superalgebra, which represents the observables that can be measured by an internal observer.

In summary, the infinite Lie algebra of observables and its cohomology are mathematical tools that can be used to study the algebra of observables and their representations in quantum field theory, including the concept of an internal observer superalgebra. The cohomology of $L_{\infty}$ with a representation provides information about the algebraic structure of the observables and can be used to study the properties of an internal observer superalgebra.

\section{Discussion}
In this work we have studied a concept of observers in a context of quantum mechanics and relativity theory introducing in particular new concepts of external and internal observers. Inspired by the achievements of algebraic topology we proposed to engage simplical complexes and sheaf cohomology theory to reconstruct overall structure of the internal space. We formulated the consistency of the internal observer as a principle of self-consistency, which ensures the consistency of information along closed loops over the internal state space of the observer. We constructed a hierarchy of external observers linked to the boundary of the system, where measurements and interactions can occur. The hierarchy covers the internal observer in a limit. It is proved that higher-level observers are equally or more non-local than lower-level observers. 
While these new concepts will require further development of formal representation, we have demonstrated their applicability to various fields, including quantum mechanics, relativity theory, algebraic quantum field theory, and loop quantum gravity.

We believe that observability is a fundamental aspect of physics that requires further studies to address current challenges, especially on the edge of quantum and relativity theories.
We anticipate that the concepts presented in this paper will have a significant impact on the future development of quantum gravity and the foundations of physics. Moreover, they may also be relevant to studies on the role of observability in artificial intelligence and models of consciousness.

\section{Acknowledgments}
Acknowledgments to Eliahu Cohen, Marek Czachor, David Chester, Mike Ebstyne, P. Horodecki, Masoud Mohseni, Daniel Sheehan and Xerxes Arsiwalla for useful discussions and comments.

\section{Appendix}\label{Appendix}

\subsection{Quantum Principle of Relativity and the Internal Observer}\label{Appendix1}

In this section we elaborate on the concept of quantum principle of relativity proposed in \cite{Dragan} and how some of its findings can be interpreted in the context of the internal observer. It is argued \cite{Dragan} that quantum randomness appearing in quantum theory, which respects relativity and makes instantaneous signaling impossible, is due to the full mathematical structure of the Lorentz transformation including the superluminal part. If one retains the superluminal term, then the key observation is that a particle moving along a single path must be abandoned and replaced by a propagation along many paths resembling the quantum mechanical approach. This implies the emergence of non-deterministic dynamics, complex probability amplitudes, and multiple trajectories.

The authors prove that relativistic invariance and symmetry requirements lead to the characterization of the
probability-like quantities that are based on a sum of complex exponential functions that we call probability
amplitudes. For the inertial frame $\{x',t'\}$ moving with the velocity $V$ relative to the frame $\{x,t\}$, there are two categories of inertial observers considered: the subluminal and superluminal. It is concluded that no relativistic, local and deterministic description of the emission of
a superluminal particle is possible in any inertial frame. If such an emission was to take place, it would have to
appear completely random to any inertial observer. The superluminal infinitely fast observer leads to apparently non-local reality representation with multiple potential evolution paths.  

Our approach, as described in this paper, is that although there is a whole hierarchy of observers, there exist two key families of observers that can be applied to this discussion with reference frames:

1. \textbf{External observers} that can be associated with subluminal reference frames $V<c$:

\begin{equation*}
    x'=\frac{x-Vt}{\sqrt{1-V^2/c^2}} 
    \end{equation*}
    \begin{equation}   
    t'=\frac{t-Vx/c^2}{\sqrt{1-V^2/c^2}}
\end{equation}

2. \textbf{Internal observers} that can be associated with superluminal reference frames $V>c$:

\begin{equation*}
    x'=\pm\frac{V}{|V|}\frac{x-Vt}{\sqrt{1-V^2/c^2}} 
    \end{equation*}
    
    \begin{equation}   
    t'=\pm\frac{V}{|V|}\frac{t-Vx/c^2}{\sqrt{V^2/c^2-1}}
\end{equation}
Interestingly for $V\rightarrow +\infty$, one gets  $\{x',t'\}=\{ct,x/c\}$, 
thus primed and non-primed frames are formally identical. This infinitely fast moving frame can be associated with the internal observer, for which local reality, local space-time and causality does not make sense. And decaying process discussed in \cite{Dragan} cannot be described by local deterministic theory. 

At this stage, following also results of \cite{Dragan}, both solutions preserve the constancy of the light speed. 

\subsection{The Consistency Condition for CTCs, P-CTCs and Entangled  Histories}\label{Appendix2}

We review in this section the consistency condition for Deutsch' model of CTCs (Closed Time-like Curves)\cite{Deutsch}, P-CTCs (Probabilistic-CTCs) \cite{Lloyd} and for the entangled histories \cite{Cot1, Cot2, NowakowskiAIP} which are an entangled variant of the quantum consistent histories \cite{Griffiths}. Readers interested in details of the aforementioned models are invited to reach out for the source literature. These consistency conditions can be viewed as a lower-level version of the self-consistency condition for the internal observer and are also an inspiration for the proposed condition of an internal observer.

The Deutsch model \cite{Deutsch} of a closed time-like curve (CTC) is a theoretical framework that explores the possibility of time travel and its potential implications on causality and consistency in quantum mechanics. 

In the Deutsch model, a CTC is a closed path in spacetime that connects an event in the past with another event in the future. The model assumes that information can be sent back in time along the CTC, allowing for the possibility of paradoxes and violations of causality.

To avoid such paradoxes, the Deutsch model introduces a consistency condition: 
\begin{equation}
\rho_{CTC}=Tr_{sys}[U(|\Psi\rangle\langle\Psi|\otimes \rho_{CTC})U^{\dag}]    
\end{equation}

This condition requires that any computation that is performed on a quantum state traveling along a CTC must be consistent with the laws of quantum mechanics. In other words, the output $\rho_{CTC}$ of any computation (modelled as an unitary operation acting on $\rho_{CTC}$ with some ancillary system in state $|\Psi\rangle$ ) cannot be in contradiction with the input, and the overall evolution of the quantum state must be consistent with the principles of unitarity.

The consistency condition can be understood as a requirement that any computation performed on a quantum state traveling along a CTC must be self-consistent and free of contradictions. This is necessary to avoid situations in which an observer could obtain information from the future that contradicts their own past experience or knowledge.

The consistency condition is a fundamental requirement in the Deutsch model of CTCs, and it is essential for ensuring that time travel does not lead to paradoxes or violations of causality in the quantum realm.

The consistency condition for a P-CTC (probabilistic closed timelike curve) was proposed by Lloyd \cite{Lloyd} as a modification of the original Deutsch model of CTCs.

In the P-CTC model, a quantum particle can travel back in time along a closed loop in spacetime, but the information it carries is probabilistic rather than deterministic. This means that there is no guarantee that the particle's past and future states will be related in a consistent manner, and the probability of obtaining a given outcome may depend on the observer's choice of measurement.

To ensure that the P-CTC model is self-consistent, Lloyd et al. introduced a modified version of the consistency condition by Deutsch:
\begin{equation}
    \mathcal{N}[\rho]=\frac{C\rho C^{\dag}}{Tr[C\rho C^{\dag}]}
\end{equation}
where $C\equiv Tr_{CTC}[U]$.
This condition requires that any measurement or computation performed on a quantum state traveling along a P-CTC must be statistically consistent with the laws of quantum mechanics. In other words, the probabilities of obtaining different outcomes must be consistent with the principles of unitarity, and the overall evolution of the quantum state must be consistent with the laws of causality.

The key difference between the consistency condition for a P-CTC and a traditional CTC is that the former allows for probabilistic outcomes rather than deterministic ones. This reflects the fact that the particle traveling along the P-CTC can potentially interact with other particles in the past or future, leading to non-deterministic behavior.
As proposed by Lloyed et al., one can consider a measurement that can be made either on the state of the system as it enters the CTC,
or on the state as it emerges from the CTC. Deutsch
demands that these two measurements yield the same
statistics for the CTC state alone: that is, the density
matrix of the system as it enters the CTC is the same as
the density matrix of the system as it exits the CTC. By
contrast, the P-CTC model demands that these two measurements yield the same statistics for the CTC state together with its
correlations with any chronology preserving variables.

The consistency condition for a P-CTC is important for ensuring that the model remains self-consistent and free of paradoxes or violations of causality. By requiring that any measurements or computations performed on a quantum state traveling along a P-CTC must be statistically consistent with the laws of quantum mechanics, the condition helps to ensure that the model is physically meaningful and does not lead to contradictions or inconsistencies.

The consistent histories interpretation of quantum mechanics is an approach to understanding the behavior of quantum systems based on the idea that a complete description of a system requires specifying its entire past and future. In this interpretation, a "history" of a quantum system is a complete specification of its state at different times, and a consistent set of histories is one that is free of contradictions and inconsistencies.

In the context of entangled histories, a consistent set of histories is one that can be constructed for a system in which particles are entangled over time. The consistency condition for consistent entangled histories requires that the histories of the individual particles be consistent with each other, even when they are separated in space and time.

The predecessor of the entangled histories is the decoherent histories approach built on the grounds of the well-known Feynman’s path integral theory for calculation of probability amplitudes of quantum processes. The entangled histories formalism extends the concepts of the consistent histories theory by allowing for complex superposition of histories. A history state is understood as an element in $\text{Proj}(\mathcal{H})$, spanned by projection operators from $\mathcal{H}$ to $\mathcal{H}$, where  $\mathcal{H}=\mathcal{H}_{t_{n}}\odot...\odot\mathcal{H}_{t_{1}}$.
The $\odot$ symbol, which we use to comply with the current literature, stands for sequential tensor products, and has the same meaning as the above $\otimes$ symbol.
The alternatives at a given instance of time form an exhaustive orthogonal set of projectors  $\sum_{\alpha_{x}}P_{x}^{\alpha_{x}}=\mathbb{I}$ and for the sample space of entangled histories $|H^{\overline{\alpha}})=P_{n}^{\alpha_{n}}\odot P_{n-1}^{\alpha_{n-1}}\odot\ldots\odot P_{1}^{\alpha_{1}}\odot P_{0}^{\alpha_{0}}$ ($\overline{\alpha}=(\alpha_{n}, \alpha_{n-1},\ldots, \alpha_{0})$), there exists a set of $c_{\overline{\alpha}} \in \mathbb{C}$ such that $\sum_{\overline{\alpha}}c_{\overline{\alpha}}|H^{\overline{\alpha}})=\mathbb{I}$ and $\sum_{\overline{\alpha}} |c_{\overline{\alpha}}|^2=1$. The family of consistent histories meets the consistency condition: $(H^{\overline{\alpha}}|H^{\overline{\beta}})=\delta_{\alpha\beta}$. The detailed discussion of the scalar product for the history vectors and consistency conditions can found in \cite{MTSEH}.

The consistency condition for consistent entangled histories requires that the histories of evolving systems are consistent, even when they are separated in space and time. This means that the outcomes of measurements on one particle must be compatible with the outcomes of measurements on the other particle, regardless of the separation between them.

\subsection{Quantum Sheaf Cohomology on Simplicial Complexes}\label{Appendix3}
In this appendix, we introduce the foundational concepts necessary for understanding the key results presented in this paper.

\begin{defn}
Let $K$ be an oriented simplicial complex and let $\sigma = (v_0,\ldots, v_k)$ be an oriented
k-simplex in $K$. For each $i \in (0, 1, \ldots , k)$, the i-th face of $\sigma$ is the $(k-1)$-dimensional simplex:
\begin{equation}
\sigma_{-i}=\{v_{0},\ldots,v_{i-1},v_{i+1},\ldots,v_{k}\}
\end{equation}
obtained by removing the i-th vertex.
\end{defn}

The boundary operator $\partial$ maps a k-simplex to a (k-1)-simplex:

\begin{defn}
Let $\sigma$ be a k-dimensional oriented simplex. The algebraic boundary of $\sigma$ is
the linear combination:
\begin{equation}
\partial\sigma=\sum_{i=0}^{k}(-1)^i\sigma_{-i}
\end{equation}
    
\end{defn}
The co-boundary operator $\delta$ is a dual operation to the boundary and maps a (k-1)-simplex to a k-simplex.

\begin{defn}
Let $K$ be a simplicial complex, a filtration $\bold{F}$ of $K$ (of length $n$) is a nested
sequence of sub-complexes of the form:
\begin{equation}
    \bold{F}_{1}K \subset \bold{F}_{2}K \subset \ldots \subset \bold{F}_{n-1}K \subset \bold{F}_{n}K =K
\end{equation}
\end{defn}

It is also worth defining the inclusion simplicial map by $g_i:\mathbf{F}_{i} K \rightarrow \mathbf{F}_{i+1} K.$

It is interesting also to recall representation of manifolds in terms of simplicial complexes: 

\begin{defn}
A manifold $\mathcal{M}$ is a simplicial complex $\triangle$ whose geometric realization $|\mathcal{M}|$ is a compact and
connected $n$-dimensional manifold.    
\end{defn}
Below we give a definition of a graph:

\begin{defn}
A graph G is a pair $(V,E)$ consisting of a set of vertices $V=\{v_i\}$ and a set of edges $E=\{e_{ij}\}$. 
\end{defn}

One can construct also a concept of nested fractal structures.
To construct a nested fractal graph structure, we start with a graph $\mathcal{G}$ and associate each vertex $v$ with its own graph $\mathcal{G}v$. This process can be iterated infinitely for each vertex of $\mathcal{G}v$, generating a sequence of graphs ${\mathcal{G}v, \mathcal{G}{v_1}, \mathcal{G}{v_2},...}$ where each successive graph is a scaled-down version of the previous one. Specifically, each graph $\mathcal{G}v$ is a fractal graph with self-similarity properties, such that the same graph can be found at different scales within $\mathcal{G}\infty$. The union of all these graphs gives rise to a nested fractal graph structure $\mathcal{G}\infty$.

Graphs are formally special examples of more general concept of simplical complexes. Thus, we give below more general definition of sheafs on simplical complexes and their cohomology group that can be generalized for representation of the internal observables superalgebra. \textbf{As stated in the paper, we associate with each level of complexity of the simplical complex an algebra of observables}.

\begin{defn}
A sheaf $\mathcal{F}$ on a graph $G$ consists of a vector space $\mathcal{F}(v)$ for each vertex $v\in G$, a vector space $\mathcal{F}(e)$ for each edge $e\in G$, and a linear transformation $\mathcal{F}_{v\leftharpoonup e}:\mathcal{F}(v)\longrightarrow \mathcal{F}(e)$ for each incident vertex-edge pair $\{v, e\}$. 
\end{defn}

Let $K$ be a simplicial complex and $\mathbb{F}$ a field. We write $(K,\leq)$ to denote the poset of simplices in $K$ ordered by the face relation.
\begin{defn}
   A sheaf over $K$ is a functor $\mathcal{F}:(K,\leq)\rightarrow Vect_{\mathbb{F}}$. $\mathcal{F}$ assigns:\\
   (1) to each simplex $\sigma$ of K a vector space $\mathcal{F}(\sigma)$ called the stalk;\\
   (2) for each $\sigma \leq \sigma'$ in K, a linear map $\mathcal{F}:(\sigma\leq\sigma'):\mathcal{F}(\sigma)\rightarrow \mathcal{F} (\sigma')$ called the restriction map meeting the following axioms:\\
   (*) for every simplex $\sigma\in K$ the map $\mathcal{F}(\sigma\leq \sigma)$ is the identity map;\\
   (**) for every triple $\sigma\leq \sigma' \leq \sigma''$ in K, $\mathcal{F}(\sigma\leq \sigma'')=\mathcal{F}(\sigma'\leq \sigma'')\circ \mathcal{F}(\sigma\leq \sigma')$ .
\end{defn}

We also assume \textit{partial ordering} of the topological structure that can be associated with an internal observer. Partial ordering means that for a set a binary relation of sequencing or inclusion is introduced. Whenever a local interaction happens, the information about the interaction is also 'stored' within the internal observer. The non-local correlations (in a broad spatio-temporal sense) are more complex structures but still are stored within an internal observer. Thus, intuitively one can associate with these type of interactions and correlations a poset (topologically partially ordered set).\textit{However, it is not an objective of this paper to digress on subtleties of the algebraic topology since the concepts related to an internal observer and its internal super-algebra might require further development of known mathematical structures.} 

The below figure represents a simple structure: 

\begin{defn}
Let $\mathcal{F}$ be a sheaf on a graph $\mathcal{G}$ and $\mathcal{V}$ be a subset of vertices of $\mathcal{G}$. Then a section $\mathcal{S}$ of $\mathcal{F}$ over $\mathcal{V}$ is a choice of vectors $x_{v}\in\mathcal{F}(v)$ for each vertex $v\in\mathcal{V}$ such that if two vertices $v$ and $v'$ are linked by an edge $e$, then $\mathcal{F}_{v\leftharpoonup e}(x_{v})=\mathcal{F}_{v'\leftharpoonup e}(x_{v'})$ . 
\end{defn}

The space $\Gamma(\mathcal{F},\mathcal{G})$ of sections over all vertices in $\mathcal{G}$ is called a global section of the sheaf $\mathcal{F}$ over the graph $\mathcal{G}$.

\begin{defn}
Let $\mathcal{F}$ be a sheaf on a graph $\mathcal{G}$. A zero-dimensional co-chain is a choice of vectors $x_v\in \mathcal{F}(v)$ for each vertex of $\mathcal{G}=\{v,e\}$ and a one-dimensional co-chain is a choice of vectors $x_e\in \mathcal{F}(e)$ for each edge of the graph. The space of zero-dimensional co-chains is denoted as $C^{0}(\mathcal{G}, \mathcal{F})$ and the space of one-dimensional co-chains as $C^{1}(\mathcal{G}, \mathcal{F})$.
\end{defn}

As a natural consequence of this definition, the space of zero-dimensional cochains can be represented as a direct sum of stalks of the sheaf over all vertices of the graph:
\begin{equation}
    C^{0}(\mathcal{G}, \mathcal{F})=\bigoplus_{v \in V}\mathcal{F}(v),
\end{equation}

and the space of one-dimensional cochains as a direct sum of stalks over the edges of the graph:

\begin{equation}
    C^{1}(\mathcal{G}, \mathcal{F})=\bigoplus_{e \in E}\mathcal{F}(e),
\end{equation}

These spaces can be generalized up to $n$-dimensional cochains $C^n(\mathcal{G}, \mathcal{F})$. It is fundamental to note that the space $C^{0}(\mathcal{G}, \mathcal{F})$ is different from the space $\Gamma(\mathcal{F},\mathcal{G})$ of global sections for which the consistency condition is imposed by the linear transformations $\mathcal{F}_{v\leftharpoonup e}$. The space of global sections is also the kernel of the coboundary map $\delta$ from the space of $0$-cochains to the space of $1$-cochains (the boundary map $\partial$ is defined as a dual operation to the co-boundary map).
The following sequence forms also a cochain complex by means of the coboundary map. 
\begin{equation}
    \varnothing\rightarrow C^0(\mathcal{G}, \mathcal{F})\xrightarrow{\delta^0} C^1(\mathcal{G}, \mathcal{F})\xrightarrow{\delta^1}\cdots\xrightarrow{\delta^{n-1}} C^n(\mathcal{G}, \mathcal{F})\xrightarrow{\delta^n}\cdots
\end{equation}

\begin{defn}
For each $k\geq 0$, the k-th cohomology group of $\mathcal{G}$ with coefficients in sheaf $\mathcal{F}$ is the quotient vector space:
\begin{equation}
    \mathbf{H}^k(\mathcal{G},\mathcal{F})=\frac{\ker\delta^k}{\img\delta^{k-1}}
\end{equation}
where $\img\delta^{k-1} \subseteq \ker\delta^k$ and $\mathbf{H}^0(\mathcal{G},\mathcal{F})=\ker\delta^0$.
\end{defn}

In general case of simplical complexes $K$, we can also define \textbf{a sheaf cohomology of a simplical complex} as the total direct sum:
\begin{equation}
    \mathbf{H}(K,\mathcal{F})=\bigoplus_k\mathbf{H}^k(K,\mathcal{F})
\end{equation}

\subsection{Non-locality of the hierarchy of observers}\label{Appendix4}

The subject of non-locality is naturally related to quantum entanglement. Quantum non-locality refers to the fact that measurements made on one particle can have an instant and unpredictable effect on the state of another particle, regardless of the distance between them. This effect violates the principle of local realism \cite{LR1, LR2, LR3, LR4}, which states that the properties of an object are determined by its local environment and not by any distant influence.
The link between quantum entanglement and quantum non-locality is that entangled particles exhibit non-local behavior. This means that when two particles are entangled, measurements made on one particle can instantaneously affect the state of the other particle, regardless of the distance between them. This non-local behavior is a consequence of the entanglement between the particles and is a key feature of quantum mechanics.
Before making any statements about the non-locality of the hierarchy of observers in our representation, we need to choose an appropriate measure for quantifying non-locality. One such measure is the entanglement measure $E(\cdot)$, which can be used to quantify the degree of non-locality in a quantum system \cite{EM1, EM2, EM3}

We present below a proof of the corrolary presented in this paper about non-locality of higher-order observers, which is based on the asumption that we engage a k-simplex $\triangle_k$ and associate algebras of quantum observables with the nodes, and apply an entanglement measure functional which meets necessary conditions to be qualified as an enganglement measure \cite{EM1}. The proof is conducted by induction on the simplex:

\begin{cor}
The observers $\mathcal{O}_{ext}^{n+1}$ at the (n+1)-th level of the hierarchy ($\triangle_{n+1}$) are always more or equally non-local than the observers $\mathcal{O}_{ext}^{n}$ at the n-th level of the hierarchy ($\triangle_{n}$).
\end{cor}

\begin{proof}
Let us assume that the nodes of the simplexes are associated with quantum observables and non-locality is measured by means of the entanglement measure $E(\cdot)$ which meets the condition of monotonicity under action of any LOCC operation $\Lambda$ (LOCC stands for local operations and classical communication), i.e. $E(\triangle_{n})\geq E(\Lambda(\triangle_{n}))$. Thus, this measure does not increase under action of a LOCC. Assume that an observer $\mathcal{O}_{ext}^{n}$ at the n-th level of the hierarchy ($\triangle_{n}$) is non-local which means that $E(\triangle_{n})>0$. The proof can be conducted further by induction. It is clear that for $n=1$ where we get 1-simplex, one can consider non-local entangled quantum system (e.g. in a maximally entangled state).   Any simplex $\triangle_{n+1}$ is generated from a simplex $\triangle_{n}$ by connecting a new node to all original nodes in $\triangle_{n}$. Formally, we can reverse this operation and following quantum version, one can get from a system $\triangle_{n+1}$ to $\triangle_{n}$ tracing out one of its nodes, i.e. applying operation $Tr_{v_i}$ on a chosen $v_{i}\in\triangle_{n+1}$. Partial trace $Tr_{v_i}$ is a local operation thus it does not increase $E(\cdot)$ due to its monotonicity and in result: 
\begin{equation}
    E(\triangle_{n+1})\geq E(\triangle_{n}=Tr_{v_i}\triangle_{n+1})
\end{equation}

We can therefore conclude that the $(n+1)$th level of the hierarchy is equally or more non-local than the $n$th level.
In a trivial case one can consider a fully separable system without any non-locality at all levels of the hierarchy. 
\end{proof}

This reasoning about hierarchies built on simplical complexes lead naturally to the hierarchy of entanglement aforementioned in this paper:
\begin{equation}
    E(\triangle_n)\rightarrow E(\triangle_{n-1})\rightarrow\dots\rightarrow E(\triangle_1)
\end{equation}
where $\triangle_{n-1}=\partial\triangle_{n}$ and in general, $E(\triangle_n) \geq E(\triangle_{n-1})\geq\dots\geq E(\triangle_1)$.

\end{document}